\documentclass[letterpaper,11pt]{article}
\usepackage{times}
\usepackage[margin=1in]{geometry}
\usepackage{amssymb,amsbsy,amsfonts,amsmath,amsthm,xspace}
\usepackage{mathrsfs}
\usepackage{graphicx}
\usepackage{caption}
\usepackage{setspace}
\usepackage{enumitem}
\usepackage{subcaption}
\usepackage{multirow}
\usepackage{epstopdf}
\usepackage{epsfig}
\usepackage{hyperref}
\usepackage{url}
\usepackage[toc,page]{appendix}
\usepackage{float}
\usepackage{natbib}
\usepackage{color}
\usepackage{verbatim}

\usepackage{titlesec}

\usepackage{algorithmic, algorithm}

\newcommand{\V}[1]{\ensuremath{\boldsymbol{#1}}\xspace}

\theoremstyle{theorem}
\newtheorem{thm}{Theorem}

\newtheorem{defi}{Definition}

\newtheorem{prop}{Proposition}
\newtheorem{coro}{Corollary}

\newtheorem{rem}{Remark}

\newcommand{\e}{\mathbb{E}}
\newcommand{\p}{\mathbb{P}}

\newcommand{\ccal}{\mathcal{C}}
\newcommand{\ycal}{\mathcal{Y}}

\newcommand{\zcal}{\mathcal{Z}}
\newcommand{\bR}{\mathbb{R}}

\newcommand{\mbone}{{\mathbf{1}}}

\newcommand{\norm}[1]{\Vert{#1}\Vert}

\def\text#1{\mbox{\rm #1}}

\title{ Diffusion Source Identification on General Networks with Statistical Confidence}

\date{}
\author{Quinlan Dawkins, Tianxi Li, Haifeng Xu}

\begin{document}

\title{\textbf{ Diffusion Source Identification on Networks with Statistical Confidence}}

%
\date{}
\author{Quinlan Dawkins, Tianxi Li, Haifeng Xu\\
University of Virginia
}
	
	\maketitle 
 \begin{abstract}
Diffusion source identification on networks is a problem of fundamental importance in a broad class of applications, including rumor controlling and virus identification.  Though this problem has received significant recent attention, most studies have focused only on very restrictive settings and lack theoretical guarantees for more realistic networks. We introduce a statistical framework for the study of diffusion source identification and develop a confidence set inference approach inspired by hypothesis testing. Our method efficiently produces a small subset of nodes, which provably covers the source node with any pre-specified confidence level without restrictive assumptions on network structures. Moreover, we propose multiple Monte Carlo strategies for the inference procedure based on network topology and the probabilistic properties that significantly improve the scalability. To our knowledge, this is the first diffusion source identification method with a practically useful theoretical guarantee on general networks. We demonstrate our approach via extensive synthetic experiments on well-known random network models and a mobility network between cities concerning the COVID-19 spreading.   \end{abstract}

 \section{Introduction}

One pressing problem  today is the spreading of misinformation or malicious attacks/virus in various cyberspaces.  For example, rumors and fake news on social networks may result in many serious political, economic, and social issues \citep{vosoughi2018spread}. Viruses that spread via emails and computer communication may cause severe privacy and leakage problems \cite{newman2002email,halperin2002system,xu2016propagation}. The negative impacts stem from a few source users/locations and then spread over the social networks via a \emph{diffusion process} in such events. One crucial step to reduce the loss from such an event is to quickly identify the sources so that counter-measures can be taken in a timely fashion. 

Though early practices have been done for this important problem with motivations from various domains, systematic research on this problem only began very recently, arguably starting from the seminal work of \cite{shah2011rumors}, which proposed a \emph{rumor center} estimator that can be located by an efficient message-passing algorithm with linear time complexity. Despite the significant interest and progress on this problem in recent years \citep{shah2012finding, dong2013rooting, khim2016confidence, bubeck2017finding,yu2018rumor,crane2020inference},  many challenges remain unaddressed. First, the theoretical understanding of these methods is currently only available under very restrictive and somewhat unrealistic structural assumptions of the networks such as \emph{regular trees}. This is perhaps partially explained by the well-known computational hardness about the probabilistic  inference of diffusion process in general graphs \cite{shapiro2012finding}. Therefore, intuitive approximations have been used for general networks \citep{nguyen2016multiple,kazemitabar2020approximate}. However, such methods lack theoretical guarantees. Second, even for regular trees, the available performance guarantee is far from being useful in practice. Even in the most idealized situation of infinite regular trees, the correct probability of the rumor center is almost always below  $0.3$ \citep{shah2011rumors, dong2013rooting, yu2018rumor}.  For general graphs, as we show later, the correct rate of such a \emph{single-point} estimation method only becomes  too low to be practical. 

To guarantee higher success probability, a typical   approach, as in both machine learning theory \cite{valiant1984theory} and data-driven applied models \cite{lecun2015deep},  is perhaps to obtain more data.  However, a fundamental challenge in   diffusion source identification (DSI)   is that \emph{the problem by nature   has only one snapshot of the network information}, i.e., the earliest observation about the infection status of the network.\footnote{Since infected nodes are usually indistinguishable and equally infectious,  any additional information in  later  observations only tells us  which \emph{new} or  \emph{additional}  nodes are infected and is not helpful for us to infer the source node.}   
Therefore, compared to classic learning tasks, DSI poses a fundamentally different challenge for inference.It is the above crucial understanding that motivates our adoption of a different statistical inference technique, the confidence set. Previously systematic statistical studies adopt the confidence set approach for DSI on trees \citep{bubeck2017finding,khim2016confidence, crane2020inference}.   Though they enjoy good theoretical properties, the methods are applicable only on infinite trees.

This paper aims to bridge the gap between practically useful algorithms and theoretical guarantees for the DSI problem. We introduce a new statistical inference framework which provably includes many previous  methods \cite{shah2011rumors,nguyen2016multiple}  as special cases. Our new framework not only  highlights the drawback of the previous methods but, more importantly, also leads to the design of our confidence set inference approach with \emph{finite-sample} theoretical guarantee on \emph{any} network structures.  

As a demonstration,  consider the example of the COVID-19 spreading procedure in early 2020. Figure~\ref{fig:china-citya} shows a travel mobility network between 49 major cities in China, constructed from the two-week travel volume \citep{DVN/FAEZIO_2020,hu2020building} before the virus caught wide attention. The square nodes (21 out of 49) are all cities with at least five confirmed cases of the virus on Jan 24, 2020.  The DSI problem is:  given  only knowledge about the   mobility network and  which cities have detected a notable amount of confirmed cases (in this case, at least 5) , can we identify in which city   the virus was first detected?  
%
 \begin{figure}[ht]
\vspace{-1.1cm}
  \begin{center}
    \includegraphics[width=0.6\textwidth]{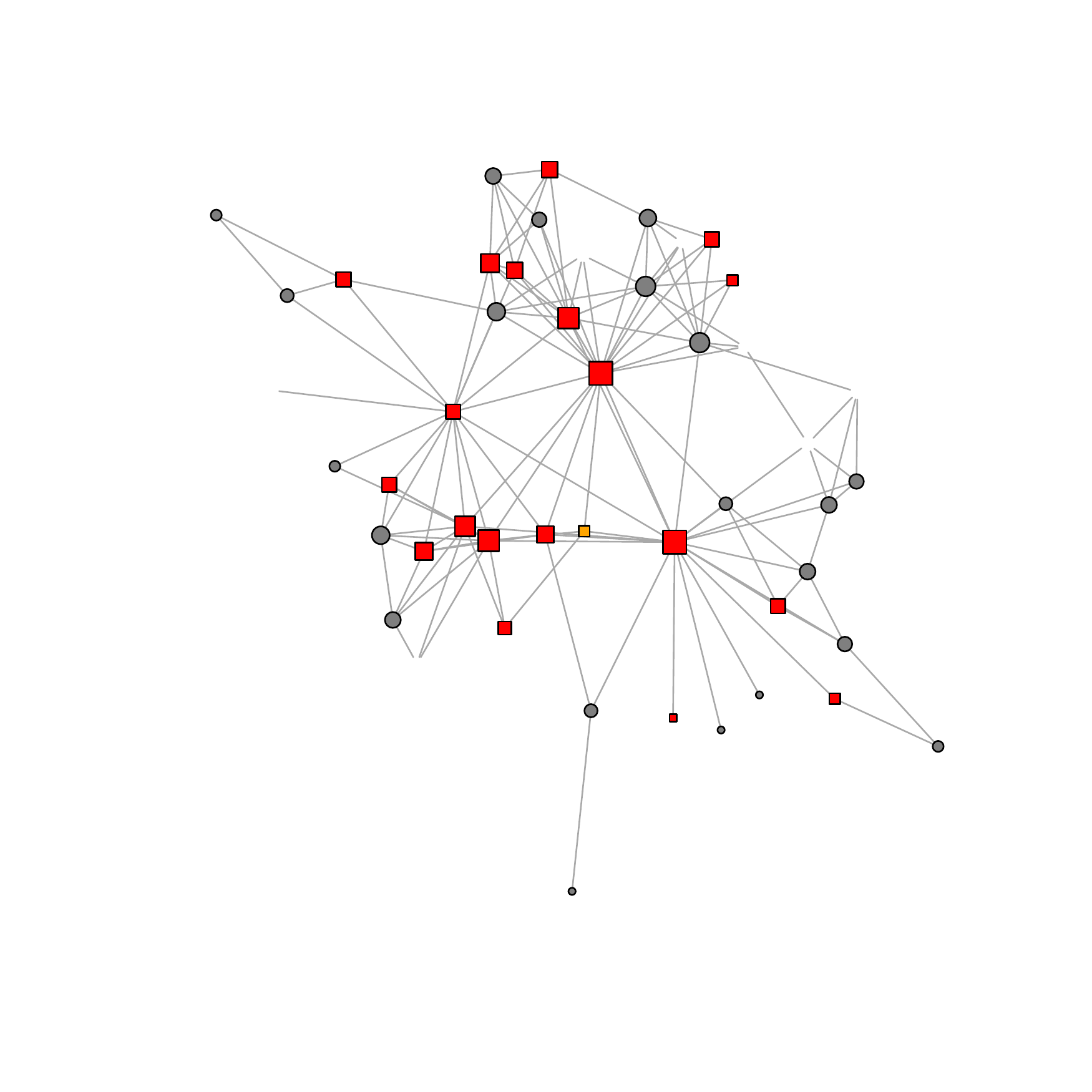}
  \end{center}
  \vspace{-1.2cm}
  \caption{ The mobility network and the COVID-19 infection status of major Chinese cities on Jan 24, 2020. Colored square nodes are cities with  at least five confirmed cases. 
  }
  \label{fig:china-citya}
\end{figure}
\vspace{-0.2cm}

This problem turns out to be too difficult for precise identification. None of the single-point source identification methods under evaluation can successfully identify Wuhan due to its relatively non-central position from the network (details in Section~\ref{sec:sim}). Nevertheless, both of our 80\% and 90\% confidence sets cover Wuhan correctly, giving recommendations of 6 nodes and 11 nodes (out of 49 cities), respectively. In fact, the evaluation on all the whole week after the lockdown of Wuhan reveals that both confidence sets correctly cover Wuhan in all the seven days, while the single-point estimation methods are rarely effective. Such a result evidently shows the necessity of adopting confidence set approach and the effectiveness of our solution.
Our   contributions in this paper can be summarized in three-folds. 
\begin{enumerate}
\item We introduce an innovative statistical framework for the DSI problem. It includes several previous methods as special cases, but has the potential for more effective inference. 	  
\item Under our framework, we propose a general way to construct the source node confidence set, whose validity can be guaranteed for  finite sample size and any network structures. It is the first DSI method with a theoretical performance guarantee on general networks, to the best of our knowledge. 	\vspace{-1mm}
\item We propose   techniques that dramatically improve the computational efficiency of our inference algorithm. En route, we develop a generalized importance sampling method, which may be of independent interest.  
\end{enumerate}
 
A high-level message in the paper is that the confidence set approach, which did not receive adequate attention in the machine learning literature, can be an important tool for inference tasks, especially for challenging problems with limited available data.

 \section{Preliminaries}
 
We start by formalizing the \emph{Diffusion Source  Identification} (DSI) problem,  introduced in the seminal work of \citet{shah2011rumors}.  Consider a network $G$  with node set $V = \{1, \cdots, n\}$ and edge set $E$. For ease of presentation, we  focus on unweighted and undirected networks but it is straightforward to generalize the model and our framework to  weighted networks.  We write $(u,v) \in E$ if node $u$ and $v$ are connected.  The network can   be equivalently represented by its $n\times n$ binary \emph{adjacency matrix} $A$, where $A_{uv} =  A_{vu}=1$ if and only if $(u, v) \in E$.

There is a \emph{source node}  $s^*\in V$  on the network $G$ initiating a diffusion of a certain effect (rumor, fake news or some virus) over the network $G$.  We embed our inference of the diffusion procedure under the widely-adopted  ``Susceptible-Infected" (SI) model \citep{anderson1992infectious,shah2011rumors}, though our approach  can be easily tailored to   other diffusion procedure as well.   In the SI model,  the source node $s^* $ is the only ``infected" node initially. The infection diffuses as follows:  given the set of currently infected nodes after $t-1$ infections, the next infection happens by  sampling uniformly at random one of the \emph{edges} connecting an infected node and a susceptible node. Consequently, a full \emph{diffusion path} with $T$ infections  can be represented by a sequence of $T+1$ nodes in the infection order. We define the \emph{diffusion path space} to be
\begin{align*}
 & \zcal_{T} = \{\V{v} = \{s^*=v_0,   v_1, \cdots, v_{T}\}:  v_t \in V,   v_{t_1} \ne v_{t_2} \\ 
 & \text{~if~} t_1 \ne t_2, \text{ and } (v_t, v_{t'}) \in E    \text{~ for some $ t < t'$ }  \}
\end{align*} 
However, in practice,  when the occurrence of the infection is noticed,   we have already lost the information about the diffusion path.   Instead,  the available  data  only contain the \emph{snapshot} of the current infection status on the network without the infection order.  Formally, the data can be represented as an $n$-dimensional binary vector $y$ with $y_i = \mbone(i \text{~is infected}) \in \{ 0, 1\}$, where $\mbone$ is the standard indicator function. Therefore, the \emph{sample space} of the DSI problem can be defined as 
\begin{align*}
& \ycal_{T} = \{y \in \{0,1\}^n:  \norm{y}_1= T, \text{such that } \{i: y_i = 1\}  \\ 
& \text{~ induces a connected subgraph of $G$}\}.
\end{align*} 
Equivalently, we will also think of any $y \in \ycal_{T}$ as the a infected subset of nodes $V_I \subset V$ with size $T$. 
The DSI problem can then be defined as follows. 
\begin{defi}[Diffusion Source Identification]
Given one sample $y \in \ycal_{T}$,  identify the source node $s^*$ of the diffusion process that generates $y$. 
\end{defi}

\noindent {\bf Challenges.} The challenge of DSI intrinsically arises  from the  loss of information in the observed data. Specifically,  by definition, we have a \emph{many-to-one} mapping $\zeta: \zcal_{T} \to \ycal_{T}$, such that $\zeta(\cdot)$ maps a  diffusion path  to the corresponding  infection snapshot of the network. Information about the infection order has been lost upon the observation of  data $y$.   Nevertheless, the DSI problem looks to identify the first node in the infection order, with access to only one snapshot of the infection status. Note that obtaining multiple snapshots  over time does not reduce the difficulty of DSI. This is because, given the current snapshot, later observed data carry no additional information about the source node due to the Markov property of the SI model. 
 
\section{A General Statistical   Framework for DSI with Confidence Guarantees}\label{sec:method}
\subsection{DSI as Parameter Estimation}
 We start by formulating  DSI under a systematic statistical framework, which will help in our design of better inference methods later on. Treating the network $G$ as fixed and $s^*$ as the model parameter,   the probability of  generating data $y \in \ycal_T$   can be represented by $\p_{s^*}(Y=y) = p(y|s^*).$ where random variable $Y$ denotes the observed data. 
The identification of $s^*$ can then be treated as a parameter estimation problem.  Specifically, we consider the following general parameter estimation framework.  Given any \emph{discrepancy function} $\ell: \ycal_T \times \zcal_T \to [0, \infty)$, we want to find an estimator of $s^*$ based on the following optimization problem:
\begin{equation}\label{eq:PE-generic-obj}
	\text{minimize}_{s}~~ \e_{s}\ell(y,Z ) 
	\vspace{-1mm}
\end{equation}
in which $Z \in \zcal_T$ is the random diffusion path following the SI model starting from parameter $s$ and $\e_{s}$ denotes the expectation over $Z$. That is, we look to select the $s$ that the diffusion path $Z$ it generates has the minimum expected discrepancy from our observed data $y$.  

\begin{rem} \emph{An important design here is that the discrepancy function $\ell$ is defined on \textbf{$\ycal_T \times \zcal_T$}, not  on $\ycal_T \times \ycal_T$. That is, $y$ will be compared with the random diffusion path while not merely the  snapshot  induced by the path.   This is because $Z$ contains richer information about the diffusion process. As we show later, this turns out to be very crucial for designing effective discrepancy functions. }
	\end{rem}

Notice that our framework include a few previous methods as special cases. Due to space limit, all formal  proofs in this paper have been deferred to the Appendix. Instead, intuition and explanations are provided as needed.  
\begin{prop}\label{prop:equivalence-rumorcenter}
\begin{enumerate}
\item	If $\ell_{rc}(y, z) = 1-\mbone(y=\zeta(z))$, when the network is an infinite regular tree, procedure \eqref{eq:PE-generic-obj} gives the rumor center of \citet{shah2011rumors}.
\item If $\ell_{se}(y, z) = \norm{y - \zeta(z)}_2^2$, the squared Euclidean distance between $y$ and $\zeta(z)$, the discrepancy is equivalent to the symmetric difference in \cite{nguyen2016multiple}\footnote{However, different from our framework, \cite{nguyen2016multiple} used an approximation metric to this discrepancy for DSI. }.
\end{enumerate}
\end{prop} 

 Proposition~\ref{prop:equivalence-rumorcenter}  also reveals some key drawbacks of the rumor center method and its variants. First, the discrepancy function $\ell_{rc}$ only takes two values, and it treats all configurations $z$ with $\zeta(z) \ne y$ equally. Therefore, such a function may not be sufficiently sensitive  for general networks. From this perspective, $\ell_{se}$ is potentially better. Second, and importantly, both of the above discrepancy functions only depend on $\zeta(z)$, failing to leverage the \emph{diffusion order} of the $z$. Ignoring such information may also undermine the performance. To overcome these drawbacks, we  propose the following family of discrepancy functions as a better alternative. We call this family the \emph{canonical family}  of discrepancy functions.
\begin{defi}[Canonical Discrepancy Functions]\label{defi:DiscrepancyFamily}
Consider a class of discrepancy functions $\ell$ that can be written in the following form
\begin{equation}\label{eq:general-loss}
\ell(y,z) = - \sum_{v: y_v = 1}\mbone(v \in z)h(t_z(v)),
	\vspace{-1mm}
\end{equation}
in which $t_z(v)$ is the infection order of node $v$ in path $z$ and $h$ is a \textbf{non-increasing weight function}. When $v \notin z$, we define $t_z(v) = \infty$.
\end{defi}
	\vspace{-0.2cm}

The canonical form \eqref{eq:general-loss} is essentially a negative similarity function. It incorporates both the infection status and the infection order of $z$. The weight function $h$ incorporates the diffusion order such that if $z$ deviates from $y$ at an early stage, the deviation is treated as a stronger signal for their discrepancy, compared with the case when they only deviates at a later stage of the diffusion. Conceptually, this canonical family is general enough to incorporate the needed information for the diffusion process. In addition, as shown in Section~\ref{sec:loss-computation}, it admits fundamental properties that make the computation  very efficient. As a special case, we demonstrate that $\ell_{se}$ is equivalent to a discrepancy function with $h(t_z(v))\equiv 2$, as follows
\begin{eqnarray*}
	\norm{y -  \zeta(z)}_2^2 = \sum_{i=1}^n\mbone(y_i \ne \zeta(z)_i) =2T - 2\sum_{v: y_v = 1}\mbone(v \in z).
\end{eqnarray*} 
Therefore $L_2$ is equivalent to Eq. \eqref{eq:general-loss} with $f(t_z(v))\equiv 2$.
 
In this paper, we are particularly interested in the following natural configuration as the discrepancy function, which we call the ``\textbf{\underline{A}veraged \underline{D}eviation - \underline{i}nverse \underline{T}ime}" (ADiT), which takes the canonical family form \eqref{eq:general-loss} with the inverse time weights:
\begin{equation}\label{eq:ADiT}
h(t_z(v)) = \frac{1}{t_z(v)}.
\end{equation}



In Table~\ref{tab:souce-point-estimation} of Section~\ref{sec:sim}, we show the simulation performance of the single-point estimation by our framework compared to other methods. Though our methods demonstrate improvements,  the accuracy is \textbf{universally low} in all situations for all methods. Such an observation indicates that it is generally impossible to recover the source node by a single estimator with high accuracy.  Indeed, as shown in \citet{shah2011rumors,dong2013rooting, yu2018rumor}, even in the ideal infinite regular tree for which the rumor center is proved to be optimal in the MLE sense, the probability of correct source node identification turns out to still be low ($\leq 0.3$).  Such a low accuracy is far from useful in real-world applications, suggesting the necessity of developing alternative forms of inference, which is we embark on in the next section. 

\subsection{Confidence Set}
As mentioned previously, single point estimators suffer from low success rates, rendering them    unsatisfactory in real-world applications. To identify the source node with a nontrivial performance guarantee, we propose constructing a small subset of nodes that provably contains the source nodes with any pre-defined confidence. This insight motivates our use of the \emph{confidence set} as the DSI method.

\begin{defi}\label{defi:confidence-set}
	Let $Y$ be the random infection status of the stochastic diffusion process starting from $s^*$. A level $1-\alpha$ confidence set of the source node is a \textbf{random} set $S(Y) \subset V$ depending on $Y$ for which
	$$\p(s^* \in S(Y))  \ge 1-\alpha.$$
\end{defi}

Surprisingly, the idea of using confidence set to infer the diffusion source --  though arguably a natural one in statistics -- has not been  explored much in the context of DSI. The most relevant  to ours are  probably  \citet{bubeck2017finding,khim2016confidence} and \citet{crane2020inference}.   \citet{bubeck2017finding} considered identifying the first node of  a growing tree but not a diffusion process. \citet{khim2016confidence} extended the idea to the SI model but only for infinite regular tree and asymptotic setting. 
Despite its theoretical merits, this method is not practical. For example, even consider the situation of an infinite 4-regular tree as the network structure, applying the method of \citet{khim2016confidence} would indicate a confidence set of size $4^{11} \approx 5\times 10^6$, regardless of the infected size $T$. This is far too large for almost any applications, let alone the fact that infinite regular tree itself is unrealistic. \citet{crane2020inference} makes the inference more effective, but still rely on the tree-structure assumption.

We instead take a completely different yet natural approach based on our statistical  framework for the problem. To ensure the validity of the inference for any network structures, we will rely on the general statistical inference strategy  for the confidence set construction. We first introduce a testing procedure for the hypothesis $H_0: s^* = s$ against the alternative hypothesis $H_1: s^* \ne s$. Given a discrepancy function $\ell$, data $y$ and the node $s$ under evaluation, define the testing statistic to be our loss $T_s(y) = \e_s\ell( y, Z)$ for any data $y$. Then the \emph{p-value} of the test is defined to be 
\begin{equation}\label{eq:p-value}
\psi_s = \p_s( T_s(\zeta(Z)) \ge T_s(y)). 
\end{equation}
where the probability $\p_s$  is over the randomness of the path $Z$  generated from the random diffusion process starting from $s$.
The p-value is the central concept in statistical hypothesis testing, and it gives a probabilistic characterization of how extreme the observed $y$ deviates from the expected range for random paths that are truly from $s$ \citep{lehmann2006testing}. For a level $1-\alpha$ confidence set, we compute $\psi_s$ for all nodes $s$ and construct the confidence set by  
\begin{equation}\label{eq:CI}
S(y) = \{s: \psi_s(y) > \alpha \}.
\end{equation}
The following result guarantees the validity of the confidence set constructed above.  
\begin{thm}\label{thm:ValidityCI}
	The confidence set constructed by \eqref{eq:CI} is a valid $1-\alpha$ confidence set.
\end{thm}
Notice that Theorem~\ref{thm:ValidityCI} is a general result, independent of the network structure or the specific test statistic we use. However, the validity of the confidence set only gives one aspect of the inference. We would like to have small confidence sets in practice since such a small set would narrow down our investigation more effectively. The confidence set size would depend on the network structure and the corresponding effectiveness of the discrepancy function (the test statistic). We will use the proposed ADiT to define our test statistic. As shown in our empirical study, it gives excellent and robust performance across various settings.

 \subsection{Algorithmic Construction of Confidence Sets}\label{secsec:algorithm}
The exact evaluation of the statistic $T_s(y)$ and p-value $\psi_s$ is infeasible for general graphs since  the probability  mass function of the SI model is intractable.  To overcome this barrier,  we resort to the Monte Carlo (MC) method for approximate calculation, with details in Algorithm~\ref{algo:generic}. This vanilla version turns out to be computationally inefficient. However, we will  introduce techniques to significantly improve its computation efficiency afterwards.   


\begin{rem}[\bf Monte Carlo setup]\emph{
Note that we have two layers of Monte Carlo evaluations. The first layer is the loss function calculation in \eqref{eq:MC-approximation} and \eqref{eq:MC-approximation-insample}, while the second layer is the p-value evaluation \eqref{eq:MC-p-value}. The first layer shares the same $m$ samples. This is different from the classical Monte Carlo, but would not break the validity for p-value calculation. The properties of p-value calculation by Monte Carlo method have been studied in detail by \cite{jockel1986finite,besag1989generalized}.
}
\end{rem}

\begin{algorithm}[ht]
	\begin{algorithmic}[1]
		\STATE {\bfseries Input:}  MC sample number $m$, confidence level $\alpha$
	   \STATE {\bfseries Input:}  Network $G$, data $y$,  discrepancy function $\ell$ 
		\FOR{ each infected node $s \in y$} 
		\STATE Generate $2m$ samples $z_i \in \zcal, i= 1, \cdots, 2m$ from the $T$-round diffusion process with source $s$ on   $G$. 
		\STATE  Estimate expected  loss $T_s(y)$ of data $y$  as
		\begin{equation}\label{eq:MC-approximation}
		\hat{T}_s(y) = \frac{1}{m}\sum_{i=m+1}^{2m}\ell_s(y, z_i).
		\end{equation}
		\STATE  For path $z_j, j = 1, \cdots, m$, estimate $T_s(\zeta(z_j))$ as 
		\begin{equation}\label{eq:MC-approximation-insample}
		\hat{T}_s(\zeta(z_j)) = \frac{1}{m}\sum_{i=m+1}^{2m}\ell(\zeta(z_j), z_i).
		\end{equation} 
		\STATE   Estimate the p-value $\psi_s(y)$ as
		\begin{equation}\label{eq:MC-p-value}
		\hat{\psi}_s(y) = \frac{1}{m}\sum_{j=1}^m\mbone(\hat{T}_s(\zeta(z_j)) \ge \hat{T}_s(y)). 
		\end{equation} 
		\ENDFOR
		\STATE {\bfseries return}   level $1-\alpha$ confidence set:
		$$\ccal_{\alpha}(y) = \{s \in V_{I}: \hat{\psi}_s(y) > \alpha\}.$$
		\caption{Vanilla MC for Confidence Set Construction}
		\label{algo:generic}
	\end{algorithmic}
\end{algorithm}

\begin{rem}[\bf Choice of the sample number $m$] \emph{
In theory, the computation in Algorithm \ref{algo:generic} is exact when $m\to \infty$. In practice, 
simple guidance about the choice of $m$ can be derived as follows. The critical step in Algorithm~\ref{algo:generic} is Step 7 for the p-value calculation since the MC errors from   previous steps are usually  in a lower order. For the correctness, we only need to worry about the evaluation at node $s^*$ when the true p-value is close to $\alpha$. Step 7  averages over $m$ indicators. By the central limit theorem, the MC estimate at most misses the true p-value by roughly $2\sqrt{\alpha(1-\alpha)/m}$. For example, if we are aiming for a 90\% confidence set where $\alpha = 0.1$, setting $m=10000$ would indicate that the MC at most misses the targeting confidence level by 0.006\%, which is usually good enough in most applications. In our experiments, we use this $m=10000$ and it has been sufficient in all situations.   Notice that this recommendation is more conservative than the ones used in classical statistical inference problems \citep{jockel1986finite}. In our experience, it might still be acceptable to use a smaller $m$.}
\end{rem}
%
%
\begin{rem}[\bf Time complexity of the vanilla MC, and its trivial parallelization]\label{rem:complexity}   \emph{The time complexity of a standard sequential implementation of Algorithm \ref{algo:generic} is $\tilde{O}(mT^2 + m^2T^2)$:\footnote{As a convention, the $\tilde{O}$ notation omits logrithmic terms. } (1) the first term is due to the MC sampling \citep{bringmann2017efficient}; (2) the second term is from the statistic calculation \eqref{eq:MC-approximation-insample} given the MC samples. However, our algorithm can be trivially parallelized. In particular, the for-loop in Step 3 can be distributed across different $s \in V_{i}$ with any communication. This leads to a parallel time complexity $\tilde{O}(mT +m^2T)$. It is worthwhile to compare this time cost with  the rumor center of \cite{shah2011rumors} which has $\tilde{O}(dT)$ linear complexity and $d$ is the maximum node degree.  But the algorithm has to be sequential (thus non-parallelizable).  In summary, Algorithm \ref{algo:generic} has a better dependence on the network density captured by $d$ but has an additional quadratic dependence on the number of samples $m$.} 
\end{rem}
\vspace{-2.5mm}


\subsection{Fast Loss Estimation  for the Canonical Family}\label{sec:loss-computation}
A major computational  bottleneck of Algorithm~\ref{algo:generic} is the $O(m^2T)$ time for  estimating $\e_s(\ell(y, Z))$ in Equation \eqref{eq:MC-approximation-insample} for every $j$ since we have to compute $\hat{\psi}$ for $m$ samples, and each $\hat{\psi}$ is the average over another $m$ samples.    Fortunately, it turns out that, for canonical discrepancy family,  this step can be done  in $O(mT)$ time, highlighting another advantage of our proposed family of cost functions.  

Instead of computing $\hat{T}_s$ in Equation \eqref{eq:MC-approximation-insample} by summing over the sample $i = m+1, \cdots, 2m$, we can  compute $\hat{T}_s$ directly using only the ``summary information'' of these samples that can be computed and cached in advance. This insight is possible due to the following alternative representation of the  $\hat{T}_s(y) $ function in  Equation \eqref{eq:MC-approximation-insample}:  
\begin{align}\label{eq:generic-loss-calculation}
\hat{T}_s(y) 
& = -\frac{1}{m}\sum_{v: y_v = 1}\sum_{i=m+1}^{2m}\sum_{k=1}^Th(k)\mbone(t_{z_i}(v)=k)\notag\\
& =  -\frac{1}{m}\sum_{v: y_v = 1}\sum_{k = 1}^TM_{v,k}h(k)
\end{align}
where $M_{v,k}$ counts the total number of   samples in $z_{m+1}, \cdots, z_{2m}$ in which  node $v$ is the $k$'th infected node in the infection path. Let   $M \in \bR^{n\times T}$ be the matrix containing the entries $M_{v,k}$. Note that, there are at most $mT$ nonzero entries in $M$ since each sample only has $T$ nodes. These entries can be computed in $O(mT)$ time simply by updating  the corresponding $M_{v,k}$ entries during sampling.  With these non-zero $M_{v,k}$ entries, we can then compute  $\hat{h}(v) = \sum_{k = 1}^TM_{v,k}h(k)$ for all the $v$ that showed up in our samples in $O(mT)$ time. Finally, given the previous  $\hat{h}(v)$, we can compute any $\hat{T}_s(y) $ in $O(T)$ time where $y =\zeta(z_1), \cdots, \zeta(z_m) $, which thus in total takes an additional $O(mT)$ time. This overall takes $O(mT)$ time.


\section{Monte Carlo Acceleration via  Pooled Sampling}\label{sec:computation}
In subsection \ref{sec:loss-computation}, we     reduced the computation time for estimating a single p-value to $\tilde{O}(mT)$, which is  arguably the minimum possible in our framework since even sampling $m$ samples already takes $\tilde{O}(mT)$.   In this section, we will introduce  efficient strategies to reduce another major computational cost in our algorithm -- the MC sampling. Our techniques will 
``borrow'' MC samples of one node for the inference task of another node by leveraging the network structure and   properties of the SI model. Consequently, we only need to generate MC samples for a subset of the infected nodes, which  may effectively reduce the computational cost.

\subsection{Surjective Importance Sampling for Single-Degree Nodes}\label{secsec:d-1node}

A node  with only one connection in the network is called a \emph{single-degree node}. Suppose node $u \in V_I$ is a single degree node with the only neighbor  $v_0$ that is also infected.   Since any diffusion process starting from $u$ must pass $v_0$, we can then use the distribution of paths from $v_0$ to infer the distribution of paths from $u$. However, the converse is not true ---  a diffusion path from $v_0$ may not pass $u$, and even if it passes $u$, this may not occur as the first infection. Therefore, certain mapping is needed to connect the two diffusion processes.   The following theorem formulates this intuition. 

\begin{thm}\label{thm:neighborhood-path}
Let $u$ be a single-degree node in the graph $G$ with the only neighbor node $v_0$. 	If a path $z \in \zcal_{T}$ starting from $v_0$ contains $u$
$$z = \{v_0, s_1, s_2, \cdots, s_{K-1}, u, s_{K+1}, \cdots, s_T\},$$
	define $z$'s  \textbf{matching path} from $u$ as 
\begin{equation}\label{eq:f-def}
f_{u}(z) = \{u, v_0, s_1, \cdots, s_{K-1}, s_{K+1}, \cdots, s_T\}.
\end{equation}
In this case, the likelihood ratio between $z$ and $f_u(z)$   is
	\begin{align}\label{eq:pmf-factor-thm} \nonumber 
		\frac{p\left(f_{u}(z)|u\right)}{p\left(z|v_0\right)} &= \frac{1}{\p\left(u|v_0,s_{1} \cdots s_{K-1}\right)} \\
		& \times \frac{1}{\prod_{k=1}^{K-1}\left(1-\p\left(s_{k}|v_0, s_1 \cdots s_{k-1}\right)\right)} 
	\end{align}
	If the path $z$ from $v_0$ that does not contain $u$, we define the ratio ${p\left(f_{u}(z)|u\right)}/{p\left(z|v_0\right)}$ to be $0$. 
\end{thm}

Notice that all terms on the right-hand side of \eqref{eq:pmf-factor-thm} are available when we sample a path from the diffusion process starting at $v_0$, thus given a sampled path $z$, computing the likelihood ratio only introduces negligible computational cost.  Intuitively, according to Theorem~\ref{thm:neighborhood-path}, when the MC samples of $v_0$ are available,  they can be used to compute the p-value for node $u$ based on a similar idea to importance sampling \citep{l2009monte}. However,  the regular importance sampling cannot be directly applied because the likelihood ratio is only available between $z$ and $f_u(z)$ under the mapping of $f_u$. Therefore, we need a generalized version of the importance sampling. We name this procedure the \emph{surjective importance sampling} and give its property in the following theorem. We believe that this theorem could be of general interest beyond our context.

\begin{thm}[Surjective Importance Sampling]\label{thm:importance-sampling}
Suppose $p_1$ and $p_2$ are two probability mass functions for discrete random vector  $Z$ defined on $\ccal_1$ and $\ccal_2$.  Let $\e_1$ and $\e_2$ denote the expectation with respect to $p_1$ and $p_2$,  respectively. Given surjection $\phi: \ccal_1' \to \ccal_2$, defined on a subset $\ccal_1' \subset \ccal_1$,   we define the inverse mapping by $\phi^{-1}(\tilde{z}) = \{z \in \ccal_1': \phi(z) = \tilde{z}\}$ for any $\tilde{z} \in \ccal_2$. For a given bounded real function of interest, $g$, define 
$$\eta  = \e_2[g(Z)] \text{~~and~~} \hat{\eta} = \frac{1}{m}\sum_{i=1}^m\frac{g(\phi(Z_i))}{|\phi^{-1}(\phi(Z_i))|}\frac{p_2(\phi(Z_i))}{p_1(Z_i)}$$
where $Z_1, Z_2, \cdots, Z_m$ is a size-$m$ i.i.d. sample from distribution $p_1$, and if $Z_i \not \in \ccal_1'$, we define $p_2(\phi(Z_i)) = 0$.
We have
$$\lim_{m\to \infty}\hat{\eta}=\eta \text{~~~~~~a.s.}$$
\end{thm}

Notice that the standard importance sampling is a special case of Theorem~\ref{thm:importance-sampling} when $\phi$ is the identity mapping. Theorem~\ref{thm:neighborhood-path} and \ref{thm:importance-sampling} toghether would serve as a cornerstone for our use of the MC samples from $v_0$ to make inference of $u$.

\begin{coro}\label{coro:single-degree-importance-sampling}
For a single degree node $u$ and its neighbor $v_0$, let $z_i, i=1, \cdots, m$ be the $m$ i.i.d. paths generated from the diffusion process with source $v_0$. For any bounded function $g$, we have
$$\lim_{m\to \infty}\frac{1}{m}\sum_{i=1}^mg\left(f_u(z_i)\right)\frac{\p\left(f_{u}(z_i)|u\right)}{\p\left(z_i|v_0\right)}\frac{1}{T} = \e_u[g(Z)] \text{~~~~~~a.s.}$$
in which $f_u(z_i)$ and the likelihood ratio is given by Theorem~\ref{thm:neighborhood-path}.
\end{coro}

Based on Corollary~\ref{coro:single-degree-importance-sampling}, when $g(z) = \ell(y, z)$ or $g(z) = \mbone(T_u(\zeta(z)) \ge T_u(y))$, $\e[g]$ corresponds to the test statistic $T_u(y)$ or the p-value $\psi_u(y)$. Consequently, the MC sampling for $u$ can be avoided. Instead, to find the p-value for $u$, Equation \eqref{eq:MC-approximation-insample}  in Algorithm~\ref{algo:generic} can be replaced by
$\hat{T}_u\left(\zeta\left(f_u\left(z_j\right)\right)\right) $ equalling the following
$$\frac{1}{m}\sum_{i=m+1}^{2m}\ell\left(\zeta\left(f_u\left(z_j\right)\right), f_u(z_i)\right)\frac{\p\left(f_{u}(z_i)|u\right)}{\p\left(z_i|v_0\right)}\frac{1}{T}$$ 
and Equation \eqref{eq:MC-p-value} can be replaced by $\hat{\psi}_u(y)$ equalling the following
$$
 \frac{1}{m}\sum_{i=1}^m\mbone\bigg(\hat{T}_u\left(\zeta\left(f_u(z_j)\right)\right) \ge \hat{T}_u(y) \bigg) \frac{\p\left(f_{u}(z_i)|u\right)}{\p\left(z_i|v_0\right)}\frac{1}{T},
$$
 
where $z_j, j = 1, \cdots, 2m$ are the MC samples generated from $v_0$. The same operation can be used for $\hat{T}_u\left(y \right)$). The computational strategy for canonical discrepancy functions can also be extended in this setting  (see Appendix~\ref{appendix:integral-importance}).
%

\subsection{Permuted Sampling for Isomorphic Nodes}\label{sec:isomophism}
 When the network structure is in some sense ``symmetric" for two nodes,   the inference properties of the MC samples from one node can be viewed as stochastically equivalent to the MC samples from the other node after the symmetric reflection. We call such a property   \emph{isomorphism}. Denote the node $u$'s $k$th order neighborhood-- the set of all nodes  (exactly) $k$ hops away from $u$-- by $N_k(u)$. The following definition for isomorphism rigorously formulates the aforementioned idea.
\begin{defi}\label{defi:1stIso}
Any two nodes $u,v$ in  a network are  \textbf{first-order isomorphic} if there exists a  permutation $\pi: \{1, 2, \cdots, n\} \to \{1, 2, \cdots, n\}$, such that: (1) $\pi(u) = v$;
(2) $\pi(i) = i, \text{if~~} i \notin \{ u, v\} \cup N_1(u)\cup N_1(v)$;
(3) $A = A_{\tilde{\pi},\tilde{\pi}}$, 
where $ A_{\tilde{\pi},\tilde{\pi}}$ is the resulting matrix by applying   permutation $\pi$ on  the rows and columns of $A$ simultaneously.
\end{defi}

For illustration, consider a simplified case of the isomorphism where $u$ and $v$ have exactly the same connections. In this case, $\pi$ only swaps $u$ and $v$ and remains the identity mapping for all other nodes. For this pair of $u, v$, the diffusion process properties would be the same if we swap the positions of $u$ and $v$. Definition~\ref{defi:1stIso} is more general than the above simplified case as it allows permutation to the first-order neighbors. 
Under this definition of isomorphism, the following theorem shows that we can use the MC samples from one node to make inference of its isomorphic nodes after applying the permutation.

\begin{thm}\label{thm:isomorphism}If $u$ and $v$ are first-order isomorphic under the permutation $\pi$. If $Z = \{u, v_1, v_2, \cdots, v_{T-1}\}$ is a random diffusion path from source $u$. Define the permuted path 
 $$Z_{\pi} = \{\pi(u), \pi(v_1), \cdots, \pi(v_{T-1})\}.$$
Then $Z_{\pi}$ has the same distribution as a random diffusion path from source $v$. 
\end{thm}

To use the MC samples of one node to its isomorphic nodes according to Theorem~\ref{thm:isomorphism}, we need an efficient algorithm to identify all isomorphic pairs and the corresponding permutations. Directly checking Definition~\ref{defi:1stIso} is costly. To speed up the computation, we identify \emph{necessary} conditions for isomorphism in Proposition~\ref{prop:iso-screen}.

\begin{prop}\label{prop:iso-screen}
If $u$ and $v$ are first-order isomorphic, we must have $d_u = d_v$ and $D_1(u) = D_1(v)$ where $d_u$ and $d_v$ are the degrees of $u$ and $v$, $D_1(u)$ and $D_2(v)$ are the degree sequence (sorted in ascending order) of $N_1(u)$ and $N_1(v)$. Furthermore, $u$ and $v$ have the same second-order neighbor sets. That is, $N_2(u) = N_2(v)$.
\end{prop}

Based on Proposition~\ref{prop:iso-screen} we can efficiently identify isomorphism using pre-screening steps. This turns out to significantly speed up our computation. Details of the algorithm are described by Algorithm~\ref{alg:isomorphicCheck} in Appendix~\ref{app:more-iso}.  With the isomorphic relations available, we can partition the nodes into isomorphic groups. Then MC sampling is only needed for one node in each group, and the MC samples can be shared within the group according to Theorem~\ref{thm:isomorphism}. Specifically, suppose $Z_1, \cdots Z_{2m}$ are sampled from the diffusion process from $u$. If $u$ and $v$ are isomorphic with permutation $\pi$, we can use $(Z_1)_{\pi}, (Z_2)_{\pi}, \cdots, (Z_{2m})_{\pi}$ as the MC samples of $v$ in Algorithm~\ref{algo:generic}.

\begin{rem} \emph{Definition \ref{defi:1stIso} can be extended to higher-order neighborhoods, identifying more isomorphic pairs. However, the complexity of identifying such  pairs increases exponentially with the order of neighbors, which may overwhelm the saved time on the MC side. The first-order isomorphism turns out to give the most desirable tradeoff in terms of computational efficiency.} 
\end{rem}

\section{Experimental Studies}\label{sec:sim}  
In this section, we evaluate our proposed methods   on   well-studied random   network models. We generate networks from three random network models: random 4-regular trees, the preferential attachment model \citep{barabasi1999emergence} and the small-world (S-W) network model \citep{watts1998collective}. In network science, the preferential attachment model is usually used to model the scale-free property of networks that is conjectured by many as ubiquity in real-world networks \citep{barabasi2013network}. The small-world property is believed to be prevalent in social networks \citep{watts1998collective}. The network size is $N=1365$ (the size of regular tree with degree 4 and depth 6) with $T=150$. The networks are sparse with average degree below 4. The Monte Carlo size $m$ is 4000. Source node in all settings is set to be the one with median eigen-centrality to be more realistic. All reported results are averaged over 200 independent replications.

 \begin{figure*}[ht]
\vspace{-0.3cm}
  \begin{center}
    \includegraphics[width=0.95\textwidth]{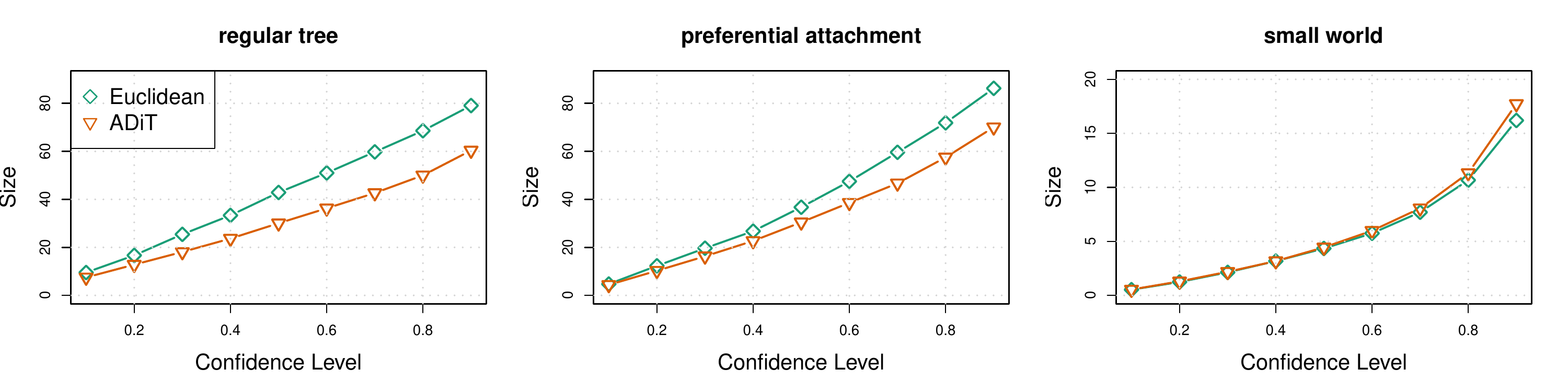}
  \end{center}
\vspace{-0.5cm}
  \caption{The average size of confidence sets for different confidence levels. }
  \label{fig:set-size}
\end{figure*}


We first evaluate the performance of the single-point source estimation accuracy from the rumor center and distance center of \cite{shah2011rumors,khim2016confidence, bubeck2017finding,yu2018rumor}, as well as estimator using our proposed framework with discrepancy functions $\ell_{se}$ and ADiT. The result is shown in the Table~\ref{tab:souce-point-estimation}.   Though the two estimators based on our framework are always better, the overall message from the table  is not promising. All of the methods, including ours, give poor accuracy that is too low to be useful in applications. Such a negative result convincingly shows that the DSI problem is generally too difficult for the single-point estimation strategy to work, and exploring the alternative confidence set inference is necessary.

\vspace{-0.05in}
\begin{table}[ht]
\caption{The correct rate of single-point estimation methods across 200 replications. }
\label{tab:souce-point-estimation}
\vspace{-0.2in}
\begin{center}
\begin{small}
\begin{sc}
\begin{tabular}{lcccc}
\hline
& reg. tree  & Pref. Att.   & S-W \\ 
\hline
Rumor center    & 0& 0 & 0.08\\
Dist. center & 0& 0&0.07\\
Euclidean (ours)  & 0& 0.04  &0.20\\
\text{ADiT} (ours)  & 0& 0.02 & 0.16\\
\hline
\end{tabular}
\end{sc}
\end{small}
\end{center}
\vskip -0.1in
\end{table}

Table~\ref{tab:souce-confidence-set} shows the coverage rate of the confidence sets, with the squared Euclidean distance and the ADiT as the discrepancy functions. Notably, the proposed confidence set procedure delivers the desired coverage in all settings (up to the simulation error). Meanwhile, the size of the confidence set varies substantially depending on the network structure. For regular trees and scale-free networks, the ADiT works much better than the Euclidean distance, indicating that the diffusion order is informative in this type of network structure. For the small-world networks, the two are very similar. This may indicate that for well-connected networks, the diffusion order is less informative. In general, we believe the adaptivity of the ADiT- based confidence set is always preferable.

\vspace{-0.05in}
\begin{table}
	\caption{The coverage rate of the confidence sets across 200 replications. }
	\label{tab:souce-confidence-set}
	\vspace{-0.2in}
	\begin{center}
		\begin{small}
			\begin{sc}
				\begin{tabular}{l|cccc}
					\hline
					& reg. tree  & Pref. Att.   & S-W \\ 
					\hline
					Euclidean-90\%  &93\% & 91\% & 92\% \\ 
					size  & 79.1 & 86.2 & 16.0 \\ 
					\hline
					\text{ADiT}-90\% & 93\% & 89\% & 91\% \\ 
					size & 60.4 & 70.5 & 17.7 \\ 
					\hline
					Euclidean-80\%  & 85\% & 84\% & 83\% \\ 
					size & 68.6 & 71.9 & 10.5 \\ 
					\hline
					\text{ADiT}-80\%  & 84\% & 82\% & 81\% \\ 
					size & 50.0& 57.5 & 11.3 \\ 
					\hline
				\end{tabular}
			\end{sc}
		\end{small}
	\end{center}
	\vskip -0.1in
\end{table}

To obtain a comprehensive view of the tradeoff between the set size and confidence level, we show the relationship between the confidence set's average size and the confidence level in  Figure~\ref{fig:set-size}. The relation is slightly sup-linear. In connection with the single-point estimation results, notice that for small-world networks, the confidence set with a confidence level 20\% has average size of around 1. In contrast, the regular tree and preferential attachment network are more difficult, and to guarantee at 10\%, the average size of the confidence set is already about 5. These observations verify the results in  Table~\ref{tab:souce-point-estimation} and support our argument that, in general, inferring the source by a single-point estimator is hopeless.

\begin{table}
\caption{The timing comparison of the sequential running time for the proposed pooled MC strategies (in sec.).}
\label{tab:timing}
\vspace{ -0.2in}
\begin{center}
\begin{small}
\begin{sc}
\begin{tabular}{lcccc}
\hline
& reg. tree  & Pref. Att.   & S-W \\ 
\hline
Vanilla MC  &898.3& 1060.5 & 1052.9 \\ 
\hline
 Import. Sampl.  & 602.3 & 462.1 & 1057.5 \\ 
\hline
Isomorphism & 617.8 & 516.8 & 1059.4 \\ 
\hline
Both & 411.9 & 430.9 & 1057.7 \\ 
\hline
\end{tabular}
\end{sc}
\end{small}
\end{center}
\vskip -0.3in
\end{table}

Finally, we also evaluate the timing improvements achieved by the pooled MC strategies. The power of the pooled MC strategies depends on network structures, as expected. The timing comparison for the pooled MC strategies is included in Table~\ref{tab:timing}. The timing included is only the sequential version of our method for a fair comparison with the rumor center. As can be seen, with both of the pooled MC strategies used, we can reduce the timing by about 60\% for tree structure and the preferential attachment networks, but the effects on small-world networks are negligible. In applications, it is observed that network structure tends to exhibit core-periphery patterns \citep{miao2021informative}, in which the pooled MC strategies are expected to be effective.   For reference, the average timing for finding the rumor center is about 1.25 seconds. However, with the extra computational cost,  our method provides a \emph{confidence sets at any specified levels}, with guaranteed accuracy for any network structures. We believe it is generally a wise tradeoff.

Meanwhile, notice that our inference procedure can be parallelized. We give a parallel algorithm in appendix (see Algorithm~\ref{algo:parallel} in Appendix~\ref{sec:parallel-algo}). It needs MC sampling for only one node in each group, and the calculations for other nodes can be done using pooled MC methods. Table~\ref{tab:timing-parallel} includes the timing results of the parallel version implementation based on 20 cores in the same settings as Table~\ref{tab:timing} of the main paper. With 20 cores, the time needed for a confidence set construction is less than 30 seconds for cases when the pooled MC methods are effective.

\begin{table}[ht]
\caption{Comparison of the parallel  running time for the proposed pooled MC strategies (in sec.) on 20 cores. }
\label{tab:timing-parallel}
\vspace{ -0.2in}
\begin{center}
\begin{small}
\begin{sc}
\begin{tabular}{lcccc}
\hline
& reg. tree  & Pref. Att.   & S-W \\ 
\hline
Vanilla MC  &59.3& 61.9 & 71.5 \\ 
\hline
 Import. Sampl.  & 33.4 & 32.9 & 73.3 \\ 
\hline
Isomorphism & 44.0& 33.4 & 73.4 \\ 
\hline
Both & 26.3 & 28.4 & 72.0 \\ 
\hline
\end{tabular}
\end{sc}
\end{small}
\end{center}
\vskip -0.3in
\end{table}

\section{Summary}\label{sec:discussion}
We have introduced a statistical inference framework for diffusion source identification on networks. Compared with previous methods, our framework is more general and renders salient insights about the problem. More importantly, within this framework, we can construct the confidence set for the source node in a more natural and principled way such that the success rate can be guaranteed on any network structure. To our knowledge, our method is the first DSI method with theoretical guarantees for general network structures. We also propose efficient computational strategies that are potentially useful in other problems as well.


\section*{Acknowledgements}
The work is supported by the Quantitative Collaborative Award from the College of Arts and
Sciences at the University of Virginia. T. Li is supported in part by the
NSF grant DMS-2015298. H. Xu is supported by a GIDI award from the UVA Global Infectious Diseases Institute. 

\bibliography{CommonBib}{}
\bibliographystyle{icml2021}

\newpage

\begin{appendix}
	\onecolumn 
%

\vspace{1cm}

\section{Proofs}

\subsection{Proof of Theorem~\ref{thm:ValidityCI}}

Define $q_{T,s^*}^{\alpha} = \inf_{t}\{t: \p_{s^*}(T_{s^*}(\zeta(Z))\ge t) \le \alpha\}$. Notice that $q_{T,s^*}^{\alpha}$ can be seen as one generalized definition for the right quantile of the distribution of the random variable $T_{s^*}(\tilde{Y})$, where $\tilde{Y}:= \zeta(Z)$ is a random infection status of the network generated by the diffusion process starting from $s^*$.
	
	Now assume $Y$ is a random infection status from the diffusion process from $s^*$. According to the definition of the p-value, we have
	\begin{align*}
		\p_{s^*}\left(s^* \in S(Y)\right) & = \p_{s^*}\left(\psi_{s^*}(Y) > \alpha\right)\\
		& = \p_{s^*}\left(  \p_{s^*}\left(T_{s^*}(\tilde{Y}) \ge T_{s^*}(Y)\right) > \alpha\right)\\
		& \ge \p_{s^*}\left(T_{s^*}(Y) < q_{T,s^*}^{\alpha}\right)\\
		& = 1-\p_{s^*}\left(T_{s^*}(Y) \ge  q_{T,s^*}^{\alpha}\right).
	\end{align*}
	Note that since $s^*$ is the true source node, $T_{s^*}(Y)$ and $T_{s^*}(\tilde{Y})$ are following exactly the same distribution, thus
	$$ \p_{s^*}\left(s^* \in S(Y)\right)  \ge 1-\p_{s^*}\left(T_{s^*}(Y) \ge  q_{T,s^*}^{\alpha}\right)  \ge 1-\alpha.$$

%

\subsection{Proof of Theorem~\ref{thm:neighborhood-path}}

Given a path generated by the diffusion process starting from $v_0$ \textbf{containing $u$}, denoted by 
$$z = \{v_0, s_1, s_2, \cdots, s_{K-1}, u, s_{K+1}, \cdots, s_T\},$$
we match it to the path $f_{u}(Z)$ defined as 
$$f_{u}(z) = \{u, v_0, s_1, \cdots, s_{K-1}, s_{K+1}, \cdots, s_T\}.$$
We start from the probability mass of $z$ starting from $v_0$. By using the Markov property, we have
\begin{align}\label{eq:pmf1}
p(z|v_0) = &\p(s_1|v_0)\p(s_2|v_0,s_1) \cdots \p(s_{K-1}|v_0,s_1, \cdots, s_{K-2}) \notag\\
 & ~~ \times \p(s_{K+1}|v_0,\cdots, u) \cdots \p(s_{T}|v_0, \cdots, s_{T-1})\\
 & ~~~~~~\times \p(u|v_0,s_1, \cdots, s_{K-1})\notag
\end{align}

In contrast, for the path $f_{u}(z)$, we have

\begin{align}\label{eq:pmf2}
\p(f_{u}(z)|u) = &\p(v_0|u)\p(s_1|u,v_0)\p(s_2|u,v_0,s_1) \cdots \p(s_{K-1}|u,v_0,s_1, \cdots, s_{K-2}) \notag\\
 & ~~ \times \p(s_{K+1}|u, v_0,\cdots, s_{K-1})  \cdots \p(s_{T}|u,v_0, \cdots, s_{T-1})\notag\\
 =&\p(s_1|u,v_0)\p(s_2|u,v_0,s_1) \cdots \p(s_{K-1}|u,v_0,s_1, \cdots, s_{K-2}) \notag\\
 & ~~ \times \p(s_{K+1}|u, v_0,\cdots, s_{K-1})  \cdots \p(s_{T}|u,v_0, \cdots, s_{T-1}).
\end{align}

Notice that the conditional probability $\p(s_{k+1}|v_0,\cdots, u, s_{K+1}, \cdots s_{k}), k > K$ only depends on the infection status before the $k$th infection and is invariant to the infection order. This property indicates that all terms after the $K+1$th (in the second rows) of \eqref{eq:pmf1} and \eqref{eq:pmf2} are equal.

Next, we compare the terms in the first line in each of \eqref{eq:pmf1} and \eqref{eq:pmf2}. Notice that for each $k < K$, the term $\p(s_{k}|v_0,s_1, \cdots, s_{k-1})$ is uniform for each available connections given the infected nodes $v_0,s_1, \cdots, s_{k-1}$ while the term $\p(s_{k}|u,v_0,s_1, \cdots, s_{k-1})$ is uniform on all available edges given the infected nodes $u,v_0,s_1, \cdots, s_{k-1}$. The only difference in the two infected sets is on $u$. Since $u$ has only one connection to $v_0$, at each point, the number of available infecting edges is one more in the former case. Therefore, we have
$$\p(s_{k}|u,v_0,s_1, \cdots, s_{k-1}) = \frac{1}{1-\p(s_{k}|v_0,s_1, \cdots, s_{k})}\p(s_{k}|v_0,s_1, \cdots, s_{k}), k< K.$$

In addition, notice that in the third line of \eqref{eq:pmf1}, there is one extra term that does not appear in \eqref{eq:pmf2}. Combining the aforementioned three relations, we final obtain probability mass factor to be
\begin{equation}\label{eq:pmf-factor}
\frac{\p(f_{u}(\pi)|S_0=u)}{\p(\pi|S_0=v_0)} = \frac{1}{(1-\p(s_1|v_0))(1-\p(s_2|v_0,s_1)) \cdots (1-\p(s_{K-1}|v_0,s_1, \cdots, s_{K-2}))\p(u|v_0,s_1, \cdots, s_{K-1})}
\end{equation}
Moreover, if $\pi$ does not contain $u$, we set the ratio to be 0.

\subsection{Proof of Theorem~\ref{thm:importance-sampling}}
Since $Z_i$'s are a random sample from $p_1$, under the current assumption, by the strong law of large numbers, we have
$$\p\left( \hat{\eta} \to \e_1[\frac{g(\phi(Z))}{|\phi^{-1}(\phi(Z))|}\frac{p_2(\phi(Z))}{p_1(Z)} ]\right) = 1.$$
Notice that $\phi$ is a surjection. Therefore,  the term $\e_1[\frac{g(\phi(Z))}{|\phi^{-1}(\phi(Z))|}\frac{p_2(\phi(Z))}{p_1(Z)} ]$ can be rewritten as
\begin{align*}
\e_1[\frac{g(\phi(Z))}{|\phi^{-1}(\phi(Z))|}\frac{p_2(\phi(Z))}{p_1(Z)} ] & = \sum_{z \in \ccal_1}\frac{g(\phi(z))}{|\phi^{-1}(\phi(z))|}\frac{p_2(\phi(z))}{p_1(z)}p_1(z)\\
& = \sum_{z \in \ccal_1}\frac{g(\phi(z))}{|\phi^{-1}(\phi(z))|}p_2(\phi(z))\\
& = \sum_{z \in \ccal_1'}\frac{g(\phi(z))}{|\phi^{-1}(\phi(z))|}p_2(\phi(z))+\sum_{z \in \ccal_1/\ccal_1'}\frac{g(\phi(z))}{|\phi^{-1}(\phi(z))|}p_2(\phi(z))\\
& = \sum_{z \in \ccal_1'}\frac{g(\phi(z))}{|\phi^{-1}(\phi(z))|}p_2(\phi(z))\\
& = \sum_{\tilde{z} \in \ccal_2}\sum_{z: \phi(z) = \tilde{z}}\frac{g(\phi(z))}{|\phi^{-1}(\phi(z))|}p_2(\phi(z))\\
& = \sum_{\tilde{z} \in \ccal_2}\sum_{z: \phi(z) = \tilde{z}}\frac{g(\tilde{z})}{|\phi^{-1}(\tilde{z})|}p_2(\tilde{z})\\
& = \sum_{\tilde{z} \in \ccal_2}\sum_{z\in \phi^{-1}(\tilde{z})}\frac{g(\tilde{z})}{|\phi^{-1}(\tilde{z})|}p_2(\tilde{z})\\
& = \sum_{\tilde{z} \in \ccal_2}|\phi^{-1}(\tilde{z})|\frac{g(\tilde{z})}{|\phi^{-1}(\tilde{z})|}p_2(\tilde{z})\\
& = \sum_{\tilde{z} \in \ccal_2}g(\tilde{z})p_2(\tilde{z})\\
& = \e_2[g(Z)].
\end{align*}

\subsection{Proof of Corollary~\ref{coro:single-degree-importance-sampling}} 
We will use $f_u(z)$ in place of $\phi$ to apply Theorem~\ref{thm:importance-sampling}. The only remaining step is to find $|f_u(f_u^{-1}(z))|$. By the definition of $f_u$ in Theorem~\ref{thm:importance-sampling}, it is easy to see that the other $T-1$ nodes (except $u$ and $v_0$) and their order uniquely determine to the mapped path. Therefore,  $|f_u(f_u^{-1}(z))|$ would always be $T$ in this situation.

\subsection{Proof of Theorem~\ref{thm:isomorphism}}
Notice that, given $T$, the probability mass function of a diffusion path only depends on the network $A$ and the source node. Condition 1 of Definition~\ref{defi:1stIso} indicates that $Z_{\pi}$ starts from $v$. Condition 2 and condition 3 of Definition~\ref{defi:1stIso} together indicate that $Z_{\pi}$ is also a valid diffusion path. Condition 3, in particular, indicates that 
$$p_u(Z) = p_v(Z_{\pi}).$$
Now define can define $\pi^{-1}$ to be the inverse of $\pi$.  For any $\tilde{Z}$ from $v$, for the same reason, $\tilde{Z}_{\pi^{-1}}$ is also a valid diffusion path starting from $u$. In particular, we have $Z = (Z_{\pi})_{\pi^{-1}}$. Therefore, $Z_{\pi}$ has the same sample space and probability mass function as the random diffusion path from $v$.

\section{ Details for Isomorphism Identification in Section \ref{sec:isomophism}}\label{app:more-iso}

\begin{algorithm}[h] 
	\caption{Identification of First-Order Isomorphic Pairs}
	\label{alg:isomorphicCheck}
	\begin{algorithmic}
		\STATE {\bfseries Input:} Graph $G=(V,E)$
		\STATE  Initialize $L=\emptyset$ to store the list of isomorphic node pairs
		\FOR{every node $u \in V$} 
		\STATE Compute $N_{1-4}(u)$ as the set of all neighbors  of $u$ within $4$ hops 
		\STATE \quad // \texttt{ $N_{1-4}(u)$ includes all nodes that are possible to be isomorphic to $u$} 
		\FOR{$v \in N_{1-4}(u)$} 
		\IF{$d_v  ==  d_u$} 
		\STATE Compute  $D_1(u) = \{ d_{u'}: u' \in N_1(u)  \}$,  the multi-set of  degrees of nodes in $N_1(u)$
		\STATE Similarly, compute   $D_1(v)$,  the multi-set of  degrees of all one-hop neighbors  of $v$
		\IF{  $D_1(u) == D_1(v)$ }  
		\STATE Compute $\tilde{N}_2(u) = N_2(u) - N_1(u) - N_1(v) - \{ u, v\}$ 
		\STATE \quad // \texttt{ $\tilde{N}_2(u)$ contains  all (exactly) two-hop neighbors  of $u$, but with  all (exactly) one-hop neighbors  of $u,v$ removed  }    
		\STATE Compute  $\tilde{N}_2(v) = N_2(v) - N_1(u) - N_1(v) - \{ u, v\}$ 
		\IF{ $\tilde{N}_2(u) == \tilde{N}_2(v)$  }
		\STATE Do exhaustive search to check whether $u,v$ are isomorphic by enumerating all possible matchings of their neighbors, and if so, add $(u,v)$ to list $L$
		\STATE \quad // \texttt{ Hopefully, not many pairs need to go through this step} 
		\ENDIF 
		\ENDIF
		\ENDIF
		\ENDFOR
		\ENDFOR
	\end{algorithmic}
\end{algorithm}

Based on the  properties in Proposition \ref{prop:iso-screen}, Algorithm \ref{alg:isomorphicCheck} finds all isomorphic pairs and the permutations in the network. Next, we provide a proof of Proposition \ref{prop:iso-screen}. 

\begin{proof}[Proof of Proposition \ref{prop:iso-screen}] 
$d_u = d_v$ because $\pi$ gives a 1-1 mapping from $N_1(u)$ to $N_1(v)$. Furthermore,  we have $\pi(N_1(u)) = N_1(v)$. By condition 3 of Definition~\ref{defi:1stIso}, we also have $D_1(u) = D_1(v)$. 

For the last one, we can prove by contradiction. Suppose there exists a node $w$, such that $w \in N_2(u)$ but $w\notin N_2(v)$. Since $\pi(N_1(u)) = N_1(v)$ while $\pi(w) = w$, so after applying the permutation to the network, we have $\pi(w)$ disconnected from $\pi(N_1(w))$. This contradicts condition 3 of Definition~\ref{defi:1stIso}.
\end{proof}

\section{Loss Function Computation Acceleration for Surjective Importance Sampling}\label{appendix:integral-importance}

The calculation strategy for canonical discrepancy functions can also be further generalized to the weighted averaging scenario used for the single-degree nodes in Section~\ref{secsec:d-1node}.  Specifically, there we need to calculate terms like
\begin{align*}
 \frac{1}{m}\sum_{i=m+1}^{2m}\ell\left(y, f_u(z_i)\right)\frac{\p\left(f_{u}(z_i)|u\right)}{\p\left(z_i|v_0\right)}\frac{1}{T}& = -\frac{1}{m}\sum_{i=m+1}^{2m}\sum_{v: y_v = 1}\mbone(v \in f_u(z_i))h(t_{f_u(z_i)}(v))\frac{\p\left(f_{u}(z_i)|u\right)}{\p\left(z_i|v_0\right)}\frac{1}{T}\\
 & = -\frac{1}{m}\sum_{v: y_v = 1}\sum_{i=m+1}^{2m}\mbone(v \in f_u(z_i))h(t_{f_u(z_i)}(v))\frac{\p\left(f_{u}(z_i)|u\right)}{\p\left(z_i|v_0\right)}\frac{1}{T}\\
 & = -\frac{1}{m}\sum_{v: y_v = 1}\sum_{i=m+1}^{2m}\mbone(v \in f_u(z_i))[\sum_{k=1}^T\mbone(t_{f_u(z_i)}(v) = k)h(k)]\frac{\p\left(f_{u}(z_i)|u\right)}{\p\left(z_i|v_0\right)}\frac{1}{T}\\
 & =  -\frac{1}{m}\sum_{v: y_v = 1}\sum_{i=m+1}^{2m}\sum_{k=1}^T\mbone(v \in f_u(z_i))\mbone(t_{f_u(z_i)}(v) = k)h(k)\frac{\p\left(f_{u}(z_i)|u\right)}{\p\left(z_i|v_0\right)}\frac{1}{T}\\
 & = -\frac{1}{m}\sum_{v: y_v = 1}\sum_{k=1}^Th(k)[\sum_{i=m+1}^{2m}\mbone(t_{f_u(z_i)}(v) = k)\frac{\p\left(f_{u}(z_i)|u\right)}{\p\left(z_i|v_0\right)}\frac{1}{T}].
\end{align*}
 Therefore, to use this strategy in Section~\ref{secsec:d-1node}, when general MC samples from $v_0$, in addition to caching $M$, we also want to cache the matrix adjusted by the factor 
 $$M^{(v_0\to u)}_{v,k} = \sum_{i=m+1}^{2m}\mbone(t_{f_u(z_i)}(v) = k)\frac{\p\left(f_{u}(z_i)|u\right)}{\p\left(z_i|v_0\right)}\frac{1}{T}.$$
 
\section{Parallel Algorithm for Confidence Set Construction}\label{sec:parallel-algo}

As discussed in Section~\ref{secsec:algorithm}, our confidence set construction algorithm can be implemented in parallel, further boosting its speed. In the main paper, we only include the details and timing for the sequential version.  The parallelized algorithm is described in Algorithm~\ref{algo:parallel}.

\begin{algorithm}[ht]
	\begin{algorithmic}[1]
		\STATE {\bfseries Input:}  MC sample number $m$, confidence level $\alpha$, Network $G$, data $y$,  discrepancy function $\ell$ 
        \STATE Compute $S = \{g_1, g_2, \cdots, g_M\}$, the isomorphic groups for infected nodes with degree at least 2.
        \FOR{each $g \in S$}
        \STATE Extend $g$ by including all of its single-degree neighbor.
        \ENDFOR
        \FOR{ each infected isomorphic group $g \in S$  \textbf{in parallel}}
        \STATE Select any $s \in g$ with degree at least 2
		\STATE Generate $2m$ samples $z_i \in \zcal, i= 1, \cdots, 2m$ from the $T$-step diffusion process from source $s$. 
		\STATE Calculate the p-value for $s$ following \eqref{eq:MC-approximation}, \eqref{eq:MC-approximation-insample}, and \eqref{eq:MC-p-value}
                \FOR{each $v\in g$ that is isomorphic to $s$}
		\STATE Calculate $\hat{\psi}_v(y)$ according to Theorem~\ref{thm:isomorphism}.
		\ENDFOR
		\FOR{each single-degree node $v \in g$}
		\STATE{Calculate the p-value $\hat{\psi}_v(y)$ according to the surjective importance sampling in Section~\ref{secsec:d-1node}.}
		\ENDFOR
   		\ENDFOR
		\STATE {\bfseries return}   the level $1-\alpha$ confidence set:
		$$\ccal_{\alpha}(y) = \{s \in V_{I}: \hat{\psi}_s(y) > \alpha\}.$$
		\caption{Parallel Confidence Set Construction}
        \label{algo:parallel}
	\end{algorithmic}
\end{algorithm}

\end{appendix}

\end{document}